\documentclass[11pt]{article}
\usepackage[a4paper,margin=28mm]{geometry}
\usepackage[T1]{fontenc}
\usepackage{lmodern}
\usepackage{amsmath,amssymb,amsthm,graphicx,booktabs,tabularx,microtype,xcolor,hyperref}
\usepackage[numbers,sort&compress]{natbib}
\hypersetup{colorlinks=true,linkcolor=blue!45!black,citecolor=blue!45!black,urlcolor=blue!55!black}
\newtheorem{theorem}{Theorem} \newcommand{\dd}{\mathrm d}  \title{Finite-Depth Null Capture and Mass-History Sensitivity in Vaidya Spacetime}
\author{ Jorge Luis Palamara \thanks{Corresponding author: \href{mailto:jlpalamara@alumni.caltech.edu}{jlpalamara@alumni.caltech.edu}; Independent Researcher, Sacramento, California, USA; ORCID: \href{https://orcid.org/0009-0001-3175-4425}{0009-0001-3175-4425}.} \and Jorge R. Palamara \thanks{\href{mailto:jpala@ucla.edu}{jpala@ucla.edu}; Independent Researcher, Los Angeles, California, USA; ORCID: \href{https://orcid.org/0009-0004-2388-8003}{0009-0004-2388-8003}.} }
\date{}

\begin{document}
\maketitle

\begin{abstract}
We formulate a finite-depth endpoint map for past-directed null rays in dynamical black-hole spacetimes, retaining an unresolved label when no terminal is reached before a prescribed optical depth. For regular transverse terminal hits, common-flow, transversality, and separation hypotheses make the regular decision sets open and the endpoint label locally constant. In Schwarzschild spacetime, the resulting finite-depth capture circle converges to the standard critical circle as $\exp[-L/(3M)]$ when the cutoff is measured in static-observer optical depth. We then apply the construction to ingoing Vaidya spacetime. Two $C^4$ accretion profiles that coincide on a neighborhood of the observation event, and therefore have identical local mass and metric jets there, nevertheless produce distinct finite-depth critical angles. For the baseline protocol,
\[
\Delta_{\mathrm{hist}}\alpha_{45}=0.0819062\,\epsilon-0.0231767\,\epsilon^2+O(\epsilon^3).
\]
A backward adjoint gives an explicit causal kernel representing the first-order response to the mass history. The kernel is broad, peaks away from the photon-sphere dwell region, and changes between finite backgrounds even where their instantaneous masses agree. The quadratic coefficient is independently confirmed by centered nonlinear solutions and Richardson extrapolation. These results separate instability-controlled cutoff convergence from globally transported mass-history sensitivity.
\end{abstract}


\section{Introduction}
Black-hole shadows, strong-deflection lensing, and photon surfaces organize critical null propagation in stationary geometries \citep{Virbhadra2000,Falcke2000,Bozza2001,Perlick2022,Claudel2001,Cao2021,Song2023}. In dynamical settings, the Vaidya solution and its global null geometry provide a standard model of spherical mass evolution \citep{Vaidya1951,Faraoni2021,Coudray2021,Vertogradov2023}. Evolving photon spheres and shadows have been studied in Vaidya-type models, including Bonnor--Vaidya and collapsing Vaidya geometries \citep{Heydarzade2024,Roy2024}. Evolution-equation, exact conformally static, and globally defined Vaidya photon-sphere and shadow constructions are represented particularly closely by Refs.~\citep{DynamicUnderstanding,VaidyaExact,Koga2022}. These works identify dynamical critical null structures or asymptotic shadow edges. Recent generalized-Vaidya work has also examined self-similar and quasistatic photon surfaces and critical impact parameters in genuinely time-dependent backgrounds~\citep{VertogradovOvgun2026}. The present construction instead treats the critical angle as a finite-depth endpoint functional and differentiates that functional with respect to the prior mass history.

Throughout $G=c=1$, the metric signature is $(-,+,+,+)$, and the celestial sphere carries the topology induced by the positive-definite observer rest-space metric. The contribution claimed is not the general statement that accretion changes a Vaidya shadow. It is the combination of: (i) a Lorentzian finite-depth endpoint map; (ii) comparison of histories with the same mass function on a neighborhood of observation; (iii) functional differentiation of the implicit critical angle with respect to mass history; and (iv) an adjoint kernel and independently validated second-order response.

The construction is not restricted to the particular Vaidya protocol considered below. Any null-propagation problem with observer-normalized initial directions, regular terminal sets, and a monotone finite-depth parameter admits an analogous endpoint and sensitivity formulation. Potential applications include dynamical lensing, time-dependent compact-object scattering, rotating numerical spacetimes, and endpoint classifications for which asymptotic ray tracing is unavailable or operationally inappropriate.

\section{Finite-Depth Endpoint Map}
Let $(\mathcal M,g)$ be time-oriented and $C^r$, $r\ge2$. At $p$, let $u$ be future-directed and unit timelike, and define
\[
\mathbb S_p=\{n\in T_p\mathcal M:g(n,u)=0,\ g(n,n)=1\}.
\]
For $n\in\mathbb S_p$, let $\eta_n(0)=p$ and $\dot\eta_n(0)=-u+n$. The affine parameter increases toward the past along the chosen past-directed ray. Let $U\in C^{r-1}$ be future-directed and unit timelike on an open set $\mathcal O$ containing all considered segments, and set
\[
\ell(n,\lambda)=\int_0^\lambda g(U,\dot\eta_n)\,\dd s.
\]
Let $\mathcal W_A\subset\mathcal O$, $A=1,\ldots,N$, be embedded $C^{r-1}$ hypersurfaces on their regular pieces. Corners, terminal intersections, and simultaneous hits are excluded from the regular sets and assigned to the discriminant. Assume that the terminal family is finite in each common flow neighborhood, or more generally that its union is relatively closed there.

For a ray that has an attained first positive hit on $\mathcal W_A$, define
\[
\lambda_A(n)=\min\{\lambda>0:\eta_n(\lambda)\in\mathcal W_A\},
\qquad D_A(n)=\ell(n,\lambda_A(n)).
\]
The minimum notation emphasizes that nonattained infima, incomplete flow segments, corners, grazing events, and simultaneous first hits are not regular outcomes. A direction belongs to $\mathcal R_A(L)$ when $\mathcal W_A$ is its unique first terminal, the hit occurs on a regular hypersurface piece and is transverse, and $D_A(n)<L$. A direction belongs to $\mathcal R_{\mathsf U}(L)$ when the ray reaches a unique parameter $\sigma_L(n)$ satisfying $\ell(n,\sigma_L)=L$ before any terminal and remains terminal-free through $[0,\sigma_L+\delta]$ for some $\delta>0$. Set
\[
\Omega_L=\mathcal R_{\mathsf U}(L)\cup\bigcup_{A=1}^{N}\mathcal R_A(L),
\]
and define the finite-depth endpoint map
\[
\Phi_L:\Omega_L\longrightarrow\{1,\ldots,N,\mathsf U\},\qquad
\Phi_L(n)=\begin{cases}
A,&n\in\mathcal R_A(L),\\
\mathsf U,&n\in\mathcal R_{\mathsf U}(L).
\end{cases}
\]
Its complement $\mathfrak D_L=\mathbb S_p\setminus\Omega_L$ is the discriminant.

\begin{theorem}[Regular first-hit stability]
Fix $n_0\in\mathbb S_p$. Assume: (i) a celestial neighborhood $V$, numbers $\lambda_A,\delta_\lambda>0$, and a common $C^{r-1}$ geodesic-flow domain on $V\times[0,\lambda_A+\delta_\lambda]\subset\mathcal O$; (ii) a unique first hit $q_A=\eta_{n_0}(\lambda_A)\in\mathcal W_A$ and no earlier terminal; (iii) $\mathcal W_A=\{h_A=0\}$ near $q_A$, with $\dd h_A\ne0$ and
\[
\dd h_A|_{q_A}(\dot\eta_{n_0}(\lambda_A))\ne0;
\]
(iv) there exists $\varepsilon\in(0,\lambda_A)$ such that $\eta_{n_0}([0,\lambda_A-\varepsilon])$ has positive distance, in an auxiliary Riemannian metric on $\mathcal O$, from every terminal hypersurface, while within a common flow neighborhood of $\eta_{n_0}([\lambda_A-\varepsilon,\lambda_A+\delta_\lambda])$ the only terminal component present is the regular piece of $\mathcal W_A$ containing $q_A$; and (v) a strict margin $D_A(n_0)=\ell(n_0,\lambda_A)<L-\delta_L$. Then $n_0$ has a celestial neighborhood on which $\mathcal W_A$ remains the unique first terminal and its hit depth remains below $L$. Hence $\mathcal R_A(L)$ is open.
\end{theorem}
\begin{proof}[Proof sketch]
Apply the implicit-function theorem to $h_A(\eta_n(\lambda))=0$. Compact separation away from the hit neighborhood excludes earlier and competing hits. The local terminal neighborhood contains only $\mathcal W_A$, while the implicit-function theorem continues its transverse intersection with the geodesic family. Continuity of the depth preserves the cutoff margin.
\end{proof}
For a regular unresolved direction, relative closedness of the terminal union and compactness of the terminal-free segment $[0,\sigma_L+\delta]$ give positive separation. The same argument applied to $\ell-L=0$, using $\partial_\lambda\ell>0$, makes $\mathcal R_{\mathsf U}(L)$ open. Therefore $\Omega_L$ is open, $\Phi_L$ is locally constant, and
\[
\partial_{\mathbb S_p}\mathcal R_A(L)\subseteq\mathfrak D_L.
\]
The closedness of $\mathfrak D_L$ is definitional; openness and local constancy are the substantive claims. Figure~\ref{fig:endpoint} summarizes the endpoint experiment.

\begin{figure}[!htbp]
\centering
\includegraphics[width=.88\linewidth]{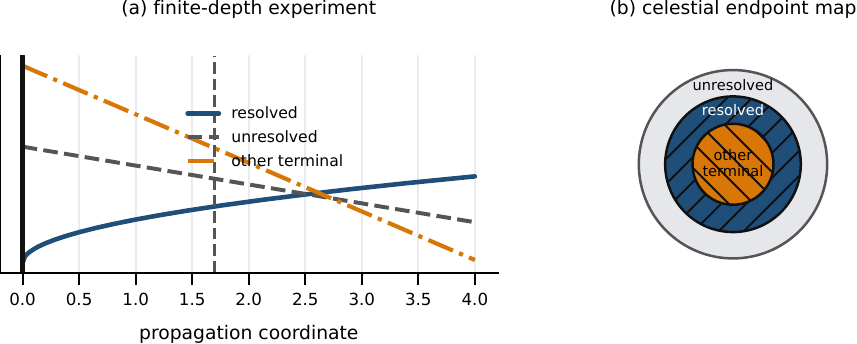}
\caption{Finite-depth endpoint construction. Regular transverse terminal hits and regular unresolved directions are locally stable; simultaneous, grazing, corner, and incomplete-flow events are assigned to the discriminant. The decision regions are distinguished by hatching, boundary styles, and direct labels.}
\label{fig:endpoint}
\end{figure}

\section{Schwarzschild Baseline}
In spherical symmetry, the terminal-resolved basin is rotationally symmetric, so its celestial boundary is a circle. For $f=1-2M/r$, a static observer at $r_o>3M$, terminal $2M<r_{\rm in}<3M$, and $U=f^{-1/2}\partial_t$,
\[
\mathcal L(b)=\int_{r_{\rm in}}^{r_o}\frac{\dd r}{\sqrt f\sqrt{1-b^2f/r^2}}.
\]
With $b_c=3\sqrt3M$, $\mathcal L$ is continuous and strictly increasing on $[0,b_c)$ and diverges at $b_c$. For $L\ge\mathcal L(0)$ there is a unique $b_L\in[0,b_c)$, with $b_L=0$ at equality; for $L>\mathcal L(0)$, one has $b_L\in(0,b_c)$. Writing $\delta=1-b^2/b_c^2$, a near/far decomposition gives
\[
\mathcal L(b)=-3M\ln\delta+C(r_o,r_{\rm in})+o(1).
\]
Thus
\[
1-\frac{b_L^2}{b_c^2}=A e^{-L/(3M)}[1+o(1)],\qquad A=e^{C/(3M)},
\]
and, for fixed $r_o>3M$, $\alpha_c-\alpha_L=K e^{-L/(3M)}[1+o(1)]$. The exponent is the photon-orbit instability rate expressed in the same depth parameter \citep{Cardoso2009}. Figure~\ref{fig:schw} shows the analytic behavior.

\begin{figure}[!htbp]
\centering
\includegraphics[width=.88\linewidth]{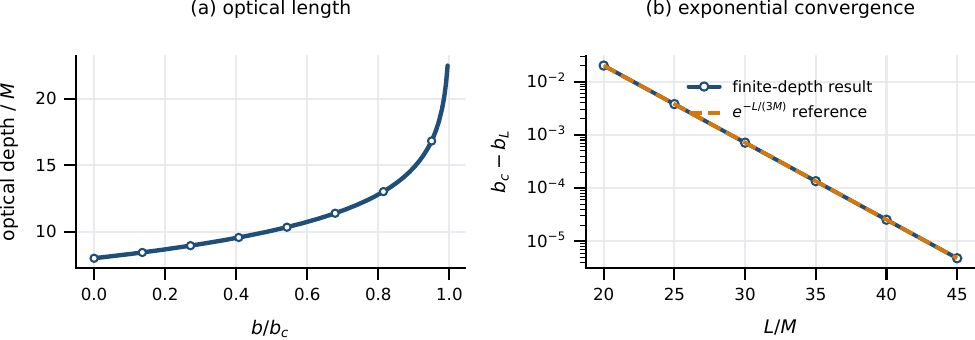}
\caption{Schwarzschild benchmark. The optical length grows monotonically below $b_c$, and the finite-depth critical impact parameter converges with the predicted $e^{-L/(3M)}$ rate. The numerical sequence and reference law are distinguished by markers and line style.}
\label{fig:schw}
\end{figure}

\section{Vaidya Endpoint Protocol}
We use
\[
\dd s^2=-f\dd v^2+2\dd v\dd r+r^2\dd\Omega^2,\qquad f=1-2m(v)/r,
\]
with Hamiltonian
\[
H=p_vp_r+\tfrac12fp_r^2+\frac{j^2}{2r^2}=0.
\]
The equations are
\[
\dot v=p_r,\quad\dot r=p_v+fp_r,\quad\dot p_v=\frac{m_v}{r}p_r^2,
\quad\dot p_r=-\frac{m}{r^2}p_r^2+\frac{j^2}{r^3},\quad\dot\phi=\frac{j}{r^2}.
\]
At observation,
\[
e_{\hat0}=f_o^{-1/2}\partial_v,\quad e_{\hat r}=f_o^{-1/2}\partial_v+f_o^{1/2}\partial_r,
\quad e_{\hat\phi}=r_o^{-1}\partial_\phi.
\]
The angle is measured from inward radial, so
\[
k=-e_{\hat0}-\cos\alpha\,e_{\hat r}+\sin\alpha\,e_{\hat\phi},
\]
and
\[
p_v(0)=\sqrt{f_o},\quad p_r(0)=-\frac{1+\cos\alpha}{\sqrt{f_o}},\quad j=r_o\sin\alpha.
\]
The sole terminal in the Vaidya calculation is the inward crossing $r=r_{\rm in}$ with $\dot r<0$; the ``other terminal'' in Fig.~\ref{fig:endpoint} illustrates only the general construction. Rays that turn beforehand are continued and are unresolved if $\ell=L$ occurs first. On the captured branch,
\[
F_{\rm root}(\alpha;m)=\ell_{\rm in}(\alpha;m)-L,
\quad r(\lambda_{\rm in})=r_{\rm in}.
\]
The selected critical angle is the capture-to-unresolved transition in $[0,\pi/2]$. For the Schwarzschild $L=45M_0$ reference, a 257-point categorical scan on $[0.25,0.65]$ rad, with spacing $0.0015625$ rad, found exactly one transition, bracketed by $[0.4828125,0.484375]$ rad. Forty-three bisection steps then reduced the bracket below $5\times10^{-14}$ rad. This is an empirical uniqueness check on the stated interval, not a global theorem. The $L=45M_0$ reference branch had $\min|\dot r|=2.89\times10^{-3}$ and $|D_\alpha|=1.09839\times10^6$.

We use $U=f^{-1/2}\partial_v$, hence $\dot\ell=p_v/\sqrt f$. The production domain obeys $r>2m(v)$. With $r_o=10M_0$, $r_{\rm in}=2.3M_0$, and $|\epsilon|\le0.08$, the following uniform analytical lower bound holds over the full protocol box:
\[
f\ge1-\frac{2(1.08)}{2.3}=0.06087>0.
\]
The profiles are
\[
m=M_0[1+\epsilon S(v;a,b)],\quad x=(v-a)/(b-a),
\]
\[
S=\begin{cases}0,&v\le a,\\126x^5-420x^6+540x^7-315x^8+70x^9,&a<v<b,\\1,&v\ge b,\end{cases}
\quad S_v:=\frac{\partial S}{\partial v}=\frac{630x^4(1-x)^4}{b-a}.
\]
This piecewise profile is $C^4$ at $v=a,b$.
The early interval is $[-20,-14]M_0$ and the late interval $[-14,-8]M_0$. The profiles coincide near $v_o=0$, so their local mass and Vaidya metric jets agree there. We define
\[
\Delta_{\rm hist}\alpha_L:=\alpha_L^{\rm early}-\alpha_L^{\rm late}.
\]
Positive values therefore mean that the early history produces the larger finite-depth critical angle. Figure~\ref{fig:protocol} shows the selected profile timing and endpoint geometry.

\begin{figure}[!htbp]
\centering
\includegraphics[width=.88\linewidth]{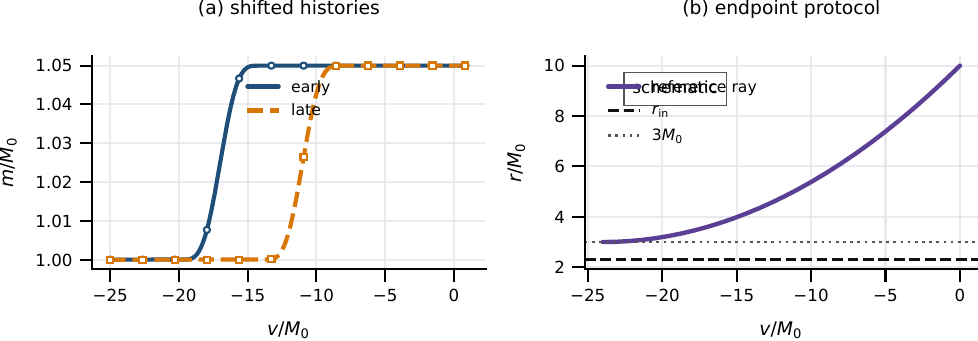}
\caption{Profiles and endpoint protocol for $r_o=10M_0$, $r_{\rm in}=2.3M_0$, early interval $[-20,-14]M_0$, late interval $[-14,-8]M_0$, and representative amplitude $\epsilon=0.05$. The early and late profiles use different line styles and markers. Panel (b) is schematic.}
\label{fig:protocol}
\end{figure}

\section{Tangent, Adjoint, and Second-Order Sensitivities}
For $z=(v,r,p_v,p_r,\phi,\ell)^T$, $\dot z=F(z;m)$ and
\[
\dot\xi_\beta=A\xi_\beta+B_\beta,\qquad A=D_zF.
\]
For $m=\bar m+\eta q$, we take $\eta$ to have dimensions of mass and $q(v)$ to be dimensionless; define $q_v:=\dd q/\dd v$. Then
\[
B_\eta=(0,-2qp_r/r,q_vp_r^2/r,-qp_r^2/r^2,0,p_vq/(rf^{3/2}))^T.
\]
At the terminal,
\[
D_\beta=\xi_\beta^\ell-\frac{\dot\ell}{\dot r}\xi_\beta^r,
\qquad \frac{\dd\alpha_L}{\dd\eta}=-D_\eta/D_\alpha.
\]
Thus $D_\beta$ is the total parameter derivative of the terminal optical depth after accounting for the parameter-dependent event time. In the $M_0=1$ normalization, $D_\alpha=\partial\ell_{\rm in}/\partial\alpha$ is quoted per radian; restoring units, it has dimensions of length. The adjoint satisfies $\dot\psi=-A^T\psi$ and
\[
\psi(\lambda_{\rm in})=(0,-\dot\ell/\dot r,0,0,0,1)^T.
\]
Write $B_\eta=bq+cq_v$, where
\[
b=(0,-2p_r/r,0,-p_r^2/r^2,0,p_v/(rf^{3/2}))^T,
\quad c=(0,0,p_r^2/r,0,0,0)^T.
\]
Before integration by parts,
\[
D_\eta=\psi(0)^T\xi_\eta(0)+\int_0^{\lambda_{\rm in}}(\psi^Tbq+\psi^Tcq_v)\,\dd\lambda.
\]
Since $\dot v=p_r<0$, define $g=\psi^Tc/p_r$. Then
\[
\frac{\dd\alpha_L}{\dd\eta}=-\frac{\psi(0)^T\xi_\eta(0)+[gq]_0^{\lambda_{\rm in}}}{D_\alpha}
+\int_{v_{\rm in}}^{v_o}K_L(v)q(v)\,\dd v,
\]
\[
K_L(v)=\left.\frac{\psi^Tb-\dd g/\dd\lambda}{D_\alpha p_r}\right|_{\lambda(v)}.
\]
For compactly supported functional perturbations the boundary term vanishes. The early-minus-late perturbation also vanishes near both endpoints. For fixed $\epsilon$, the early and late profiles induce the same observation tetrad because they have the same final mass. In an expansion about $\epsilon=0$, however, the common tetrad itself varies with $\epsilon$; its contribution cancels from the first-order early-minus-late difference but enters mixed and higher-order terms through the background-dependent flow. We plot
\[
\widehat K_L(\hat v)=M_0^2K_L(M_0\hat v),\qquad \hat v=\frac{v}{M_0}.
\]
Thus $K_L$ has dimensions of inverse length squared, and $\widehat K_L=M_0^2K_L$ is dimensionless.
We call the kernel causal in the operational sense that it has support only on the advanced-time interval sampled by the past-directed ray before the endpoint decision.

For second order, reparameterize the monotonic captured branch by $r$. For $x=(v,p_v,p_r,\phi,\ell)^T$ and parameters $p_1=\alpha,p_2=\epsilon$,
\[
X_i'=T_xX_i+T_{p_i},
\quad Y_{ij}'=T_xY_{ij}+T_{xx}[X_i,X_j]+T_{xp_i}X_j+T_{xp_j}X_i+T_{p_ip_j}.
\]
Symbolic derivatives are generated from the common radial flow. The implicit formulas are
\[
\alpha'=-D_\epsilon/D_\alpha,
\quad \alpha''=-\frac{D_{\epsilon\epsilon}+2D_{\alpha\epsilon}\alpha'+D_{\alpha\alpha}(\alpha')^2}{D_\alpha}.
\]

\section{Numerical Results and Reproducibility}
At $\epsilon=0.05$, $L=30M_0$,
\[
\alpha_{30}^{\rm early}=0.4943233,\quad\alpha_{30}^{\rm late}=0.4902946,
\quad\Delta_{\rm hist}\alpha_{30}=0.00402870\ {\rm rad}.
\]
The sign remains positive over $r_{\rm in}/M_0=2.2,2.3,2.4,2.5$ and $L/M_0=25,30,35$. At $L=45M_0$,
\[
\boxed{\Delta_{\rm hist}\alpha_{45}=0.0819062\epsilon-0.0231767\epsilon^2+O(\epsilon^3).}
\]
The quadratic approximation is within $0.037\%$ at $\epsilon=0.05$ and $0.094\%$ at $\epsilon=0.08$. Write $A_L(\epsilon):=\Delta_{\rm hist}\alpha_L(\epsilon)$. For an exponential tail $A_L=A_\infty-Ce^{-\kappa L}$ and equally spaced cutoffs with spacing $h_L=5M_0$, we eliminate $A_\infty$ and estimate
\[
\widehat\kappa_L=\frac{1}{h_L}\ln\left|\frac{A_L-A_{L-h_L}}{A_{L+h_L}-A_L}\right|.
\]
Evaluating $\widehat\kappa_{40}$ from $L/M_0=35,40,45$ gives $\widehat\kappa_{40}M_0=0.33311$, $0.33313$, $0.33311$, and $0.33314$ for $\epsilon=0.01$, $0.02$, $0.05$, and $0.08$, respectively. Figure~\ref{fig:cutoff}(b) plots the finite-reference difference $R_L(\epsilon)=|A_{45M_0}(\epsilon)-A_L(\epsilon)|$ for $L<45M_0$; it is not an estimate of $|A_\infty-A_L|$. This cutoff-rate comparison is a consistency diagnostic, not a Vaidya universality claim. Figures~\ref{fig:cutoff} and \ref{fig:amplitude} summarize the cutoff and amplitude sequences.

\begin{figure}[!htbp]
\centering
\includegraphics[width=.9\linewidth]{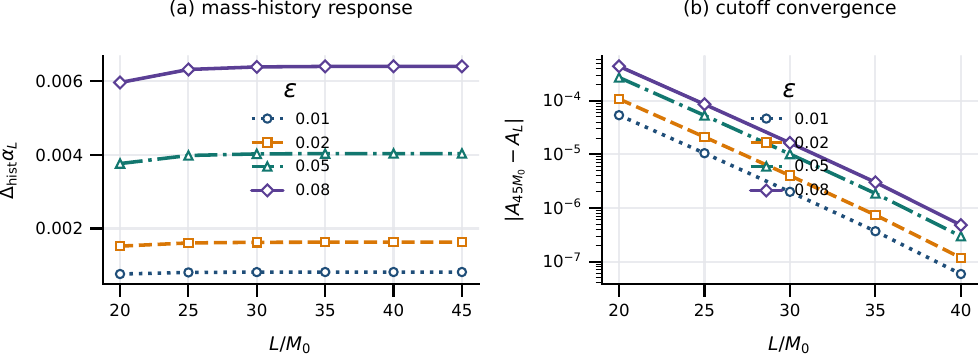}
\caption{Mass-history response and cutoff convergence for $r_o=10M_0$, $r_{\rm in}=2.3M_0$, the fixed early and late profile intervals, and amplitudes $\epsilon=0.01,0.02,0.05,0.08$. Panel (b) shows the finite-reference difference $R_L(\epsilon)=|A_{45M_0}(\epsilon)-A_L(\epsilon)|$ for $L<45M_0$; it is not an estimate of $|A_\infty-A_L|$. The four amplitudes are distinguished by line style and marker.}
\label{fig:cutoff}
\end{figure}

\begin{figure}[!htbp]
\centering
\includegraphics[width=.9\linewidth]{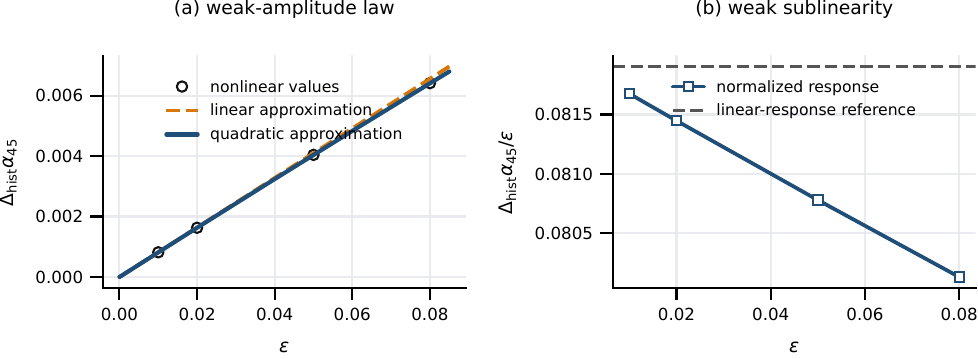}
\caption{Weak-amplitude response at $L=45M_0$. Numerical values, the linear approximation, and the quadratic approximation use distinct markers and line styles. The quadratic curve uses $c_1=0.0819062$ and $c_2=-0.0231767$.}
\label{fig:amplitude}
\end{figure}

Calculations used standard adaptive ODE techniques \citep{Hairer1993}. Second-order derivatives were generated symbolically with SymPy from the common radial flow. The implementation used Python 3.12.14, NumPy 1.26.4, SciPy 1.13.1, SymPy 1.14.0, IEEE double precision, and adaptive DOP853 integration. Root stopping tolerance was $5\times10^{-13}$ rad or tighter. Dense-output adjoint evaluation used at least 12,001 affine nodes. The negative-$\epsilon$ solutions are used only as an auxiliary analytic continuation for the centered second-derivative test. In ingoing Vaidya they correspond to $m_v<0$ on the transition interval and are not interpreted here as physical positive-energy accretion histories. The independent nonlinear check gave the values in Table~\ref{tab:centered}.
\begin{table}[!htbp]
\centering\small
\caption{Nonlinear early and late critical angles, in radians, used in the centered estimate of $c_2$.}
\label{tab:centered}
\begin{tabular}{c rr rr}\toprule
$h$&$\alpha_{\rm early}(-h)$&$\alpha_{\rm late}(-h)$&$\alpha_{\rm early}(h)$&$\alpha_{\rm late}(h)$\\\midrule
0.004&0.4824446511943&0.4827726475859&0.4842720217666&0.4839447670338\\
0.002&0.4829016550323&0.4830655602268&0.4838153395220&0.4836516197416\\
0.001&0.4831301164034&0.4832120457872&0.4835869585487&0.4835050755184\\\bottomrule
\end{tabular}
\end{table}
The corresponding $c_2(h)$ are $-0.0231768364$, $-0.0231767629$, and $-0.0231767446$; using the last pair $h=0.002,0.001$ and $c_2^{\rm Rich}(h)=[4c_2(h/2)-c_2(h)]/3$, Richardson extrapolation gives $-0.0231767385$. Representative archived tangent-window and adjoint-convolution values are shown in Table~\ref{tab:adjoint}; the complete archived pairs are supplied electronically.
\begin{table}[!htbp]
\centering\small
\caption{Archived production cross-check: direct-tangent and adjoint-convolution response values. The original perturbation-window shape, width, and normalization were not retained; the centers identify the archived comparison rows.}
\label{tab:adjoint}
\begin{tabular}{r cc}\toprule
archived window center $\hat v_c$&adjoint&direct tangent\\\midrule
$-40$&0.0016374022&0.0016374023\\
$-30$&0.0082169734&0.0082169733\\
$-20$&0.0162705878&0.0162705877\\
$-10$&0.0107509541&0.0107509542\\\bottomrule
\end{tabular}
\end{table}
The displayed rounded values differ in their final quoted decimal place; full-precision archived values are supplied electronically. These archived pairs record agreement between the production tangent and adjoint implementations at the time of calculation. Because the original driver and perturbation-window definition were not retained, the comparison is not independently reproducible and does not bound the production integrator error.

The Schwarzschild-background kernel peaks at $\hat v=-19.82$, $r=4.26M_0$, with $\widehat K_{45}^{\max}\simeq0.01630$, while the near-$3M_0$ dwell occurs near $\hat v=-54$. Linear interpolation between the supplied kernel samples gives the approximate half-maximum support $-30\lesssim\hat v\lesssim-8$. At $\epsilon=0.05$, the archived early and late peak values are approximately $0.01621$ and $0.01642$. Figure~\ref{fig:kernel} displays the archived kernel summaries.

\begin{figure}[!htbp]
\centering
\includegraphics[width=.9\linewidth]{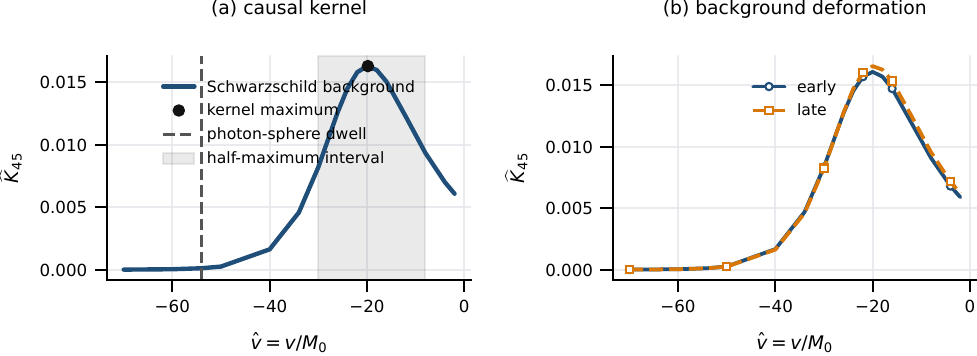}
\caption{Dimensionless mass-history kernel $\widehat K_{45}$ at $L=45M_0$, $r_o=10M_0$, and $r_{\rm in}=2.3M_0$. Panel (a) shows the Schwarzschild background, kernel maximum, photon-sphere dwell marker, and half-maximum interval. Panel (b) compares the $\epsilon=0.05$ early and late backgrounds using distinct line styles and markers.}
\label{fig:kernel}
\end{figure}
A second observer radius gives $c_1(15M_0)=0.0231730>0$, a qualitative robustness check rather than a scaling law.

\section{Discussion}
The endpoint boundary is an observer-dependent decision boundary on the celestial sphere. It is not a photon surface, which is a spacetime-geometric object, and it is not automatically a radiative image edge. Horizon-scale images additionally require an emissivity and absorption model, frequency shifts, propagation of specific intensity, and detector modeling \citep{EHT2019,EHT2022}; none of those ingredients is included here.

The unresolved outcome is essential. Failure to reach the inward terminal before a finite cutoff does not establish asymptotic escape, especially near unstable trapping where the decision time grows rapidly. Retaining an unresolved class prevents a finite computational experiment from being interpreted as an asymptotic classification.

The observable is protocol dependent. It depends on the observer, depth congruence, terminal set, cutoff, and profile family. Within that protocol, the Schwarzschild benchmark and Vaidya calculations separate two mechanisms. Local orbit instability controls the exponential convergence with cutoff, while the change between equal-local-jet histories is transported globally along the null ray and its adjoint. The causal kernel makes that transport explicit: the response is broad in advanced time, its maximum is displaced from the near-photon-sphere dwell region, and finite backgrounds deform the kernel even where the instantaneous mass agrees.

The calculation is restricted by spherical symmetry and by the use of one timelike depth congruence. It does not establish invariance under changes of the depth field, nor does it include radiative transfer. Extending the endpoint and adjoint framework to rotating backgrounds, nonspherical terminal maps, or numerically generated spacetimes would require new calculations, but no change in the logical structure of the finite-depth classification.

\section{Conclusions}
We defined a finite-depth endpoint map for observer-normalized null directions with an explicit unresolved outcome. Under common-flow, transversality, separation, relative-closedness, and strict-cutoff hypotheses, regular first-hit and unresolved decision sets are open and the endpoint label is locally constant.

In Schwarzschild spacetime, the finite-depth capture circle converges to the standard critical circle as $e^{-L/(3M)}$ in static-observer optical depth. In ingoing Vaidya spacetime, two $C^4$ mass histories that coincide near observation nevertheless yield different finite-depth critical angles. For the baseline protocol,
\[
\Delta_{\rm hist}\alpha_{45}=0.0819062\epsilon-0.0231767\epsilon^2+O(\epsilon^3).
\]
The backward adjoint represents the linear mass-history response by a causal kernel and shows that the finite-depth critical angle is a nonlocal functional of the sampled history. The independently refined second-order system and centered nonlinear Richardson check support the quoted quadratic coefficient.

These conclusions apply to the specified finite-depth endpoint protocol. They do not identify a radiative image edge, establish depth-congruence invariance, or remove the limitations of spherical symmetry.

\section*{Statements and Declarations}
\paragraph{Funding.} No funds, grants, or other support were received for conducting this study or preparing the manuscript.

\paragraph{Competing interests.} The authors have no relevant financial or non-financial interests to disclose.

\paragraph{Author contributions.} Jorge Luis Palamara and Jorge R. Palamara jointly developed the study, analyzed the results, and revised the manuscript. Jorge Luis Palamara developed the analytical and numerical methods, performed the calculations, and prepared the original draft. Both authors reviewed and approved the final manuscript.

\paragraph{Use of artificial intelligence and generative tools.} During the preparation of this manuscript, the authors used AI-assisted language tools for language refinement, structural editing, \LaTeX{} formatting, and code-syntax checking. The scientific framework, mathematical derivations, numerical calculations, data, and physical conclusions were developed and verified by the authors. All AI-assisted output was reviewed by the authors, who take full responsibility for the accuracy and integrity of the final manuscript.

\paragraph{Data and code availability.} Supplementary Information version 3.7 contains the retained cutoff, amplitude, kernel, tangent--adjoint comparison, and refinement data; executable scripts for the categorical scan, independent radial critical-angle check, cutoff-rate estimator, and figure generation; an environment manifest; and a documented build workflow. The original production tangent--adjoint driver was not retained, so the archived kernel and tangent--adjoint values are supplied as numerical data rather than as fully executable first-principles reproduction products. The code and data are released under the MIT License.

\appendix
\section{Numerical Controls}
Table~\ref{tab:controls} separates the settings that are directly recoverable from the retained scripts from the archived comparison products. The categorical scan and independent nonlinear check are executable from the Supplementary Information. The retained archive contains the tangent--adjoint comparison pairs and kernel samples, but not the original production tangent/adjoint driver; consequently the comparison in Table~\ref{tab:adjoint} should not be read as an independent bound on the production integrator truncation error.

\begin{table}[!htbp]
\centering\small
\caption{Documented numerical controls. A dash denotes a setting that is not applicable or is not recoverable from the retained executable archive.}
\label{tab:controls}
\begin{tabularx}{\linewidth}{>{\raggedright\arraybackslash}X p{0.14\linewidth} p{0.20\linewidth} >{\raggedright\arraybackslash}X}\toprule
calculation&method&tolerances&sampling, step, or event control\\\midrule
categorical reference scan&DOP853&rtol $2\times10^{-9}$; atol $2\times10^{-11}$&maximum affine step $0.15$; terminal/cutoff events; 257 angles\\
independent radial critical-angle check&DOP853&rtol $8\times10^{-12}$; atol $8\times10^{-14}$&maximum $\Delta r=0.006M_0$; secant tolerance $10^{-13}$ rad\\
tangent--adjoint comparison&archived pair table&--&12,001 sampled nodes\\
second-order radial flows S1--S3&DOP853&$(10^{-9},10^{-11})$, $(10^{-10},10^{-12})$, $(10^{-11},10^{-13})$&maximum $\Delta r/M_0=0.01,0.005,0.0025$\\\bottomrule
\end{tabularx}
\end{table}

The second-order settings S1--S3 produced $c_2=-0.0231767377262$, $-0.0231767378857$, and $-0.0231767378799$, respectively. Supplementary Information version 3.7 includes branch diagnostics, cutoff tables, archived kernel samples, direct tangent--adjoint comparison pairs, and the independent centered-difference data. A separate kernel node-refinement claim is not made because the retained archive does not contain a 6,001/12,001/24,001-node sequence.

\clearpage
\bibliographystyle{unsrtnat}
\bibliography{references}
\end{document}